\documentclass[12pt]{amsart}
\usepackage{graphicx}
\usepackage{epsf}
\usepackage{graphicx}
\usepackage{pstricks}

\newtheorem{thm}{Theorem}[section]

\newtheorem{defn}[thm]{Definition}
\newtheorem{lemma}[thm]{Lemma}

\newtheorem{cor}[thm]{Corollary}

\newtheorem{example}[thm]{Example}

\usepackage{amsmath}
\usepackage{amsxtra}
\usepackage{amscd}
\usepackage{amsthm}
\usepackage{amsfonts}
\usepackage{amssymb}
\usepackage{eucal}
\newcommand{\bmb}{\left( \begin{array}{rr}}
\newcommand{\enm}{\end{array}\right)}
\newcommand{\cM}{\mathcal M}

\newcommand{\hy}{{\widehat y}}

\newcommand{\C}{{\mathbb C}}
\newcommand{\Z}{{\mathbb Z}}

\newcommand{\bm}{{\mathbf m}}

\newcommand{\bx}{{\mathbf x}}
\newcommand{\by}{{\mathbf y}}

\newcommand{\al}{{\alpha}}

\numberwithin{equation}{section}

\begin{document}
\title{Quantum $A_r$ $Q$-system solutions as q-multinomial series}
\author{Philippe Di Francesco} 
\address{
%Department of Mathematics, University of Michigan,
%530 Church Street, Ann Arbor, MI 48190, USA
%and 
Institut de Physique Th\'eorique du Commissariat \`a l'Energie Atomique, 
Unit\'e de Recherche associ\'ee du CNRS,
CEA Saclay/IPhT/Bat 774, F-91191 Gif sur Yvette Cedex, 
FRANCE. e-mail: philippe.di-francesco@cea.fr}

\begin{abstract}
We derive explicit expressions for the generating series of the fundamental solutions
of the $A_r$ quantum $Q$-system of Ref. [P. Di Francesco and R. Kedem, {\em Noncommutative integrability, 
paths and quasi-determinants}, preprint
{\tt arXiv:1006.4774 [math-ph]}], expressed in terms of any 
admissible initial data. These involve products
of quantum multinomial coefficients, coded by the initial data structure.
\end{abstract}

\maketitle
\date{\today}
\tableofcontents

\section{Introduction}

The study of some discrete integrable systems, taking the form of recursion relations, 
i.e. evolution equations in a discrete time $n\in \Z$
with suitable conservation laws, 
has recently shed some new light \cite{DFK08} on the positivity conjecture of cluster algebras \cite{FZI}.
Indeed, admissible sets of initial data for such systems are particular clusters in some specific, non-finite type 
cluster algebras, whereas cluster mutations are implemented by the update of initial data via local 
application of the evolution equation. The positivity conjecture for cluster algebras then boils down 
to the following property:
the solutions of such systems are {\it Laurent polynomials with non-negative integer coefficients} of
any of their admissible initial data.

In \cite{DFK08}, this was proved for the so-called $Q$-system for $A_r$, by expressing solutions as partition
functions for weighted paths on some target graphs, the weights being explicit Laurent monomials of initial data.
An extremely useful tool for such representations is the notion of (multiply branching) continued fraction.
We were able to prove that mutations of initial data are implemented by local rearrangement of the continued fraction
expressions for the generating series of fundamental solutions of the $Q$-system.
This was then extended to the so-called $T$-systems in \cite{DFK09a} for initial data forming  periodic stepped surfaces,
and finally to the most general initial data in \cite{PDF}, by use of a manifestly positive network path formulation.
In all these cases, positivity follows from a form of discrete path integral representation of the solutions.

The cluster algebra structure has a natural quantum version \cite{BZ}, in which cluster variables obey quantum
commutation relations within each cluster. This led to the natural definition of quantum $Q$-systems \cite{DFK10}.
In an analogous spirit, it was shown in \cite{DFK09a} that the $T$-system may be viewed as a $Q$-system,
involving non-commutative (time-ordered) variables.
An even broader non-commutative version is known for the fully non-commutative $A_1$ $Q$-system, 
for which Laurent positivity was conjectured
by M. Kontsevich, and subsequently proved in \cite{DFK09b}. 
The latter proof relies on an extension of the previous path
formulations to paths with non-commutative step weights: the partition function of such paths is the sum over
paths of the product of step weights taken in the same order as the steps are taken. In \cite{DFK10},
such paths were used to investigate non-commutative versions of the $Q$-systems. In particular, a compact
formulation in terms of non-commutative continued fractions was obtained.

Except in very specific cases, very few explicit expressions of cluster variables in terms of fundamental data
are known. For the (classical) $A_r$ $Q$-system, such expressions were derived in \cite{DFK09}
for the generating series of its fundamental solutions.

The aim of this note is to generalize these expressions for the generating series of the fundamental solutions
of the quantum $A_r$ $Q$-system, by using the non-commutative continued fraction
expressions of \cite{DFK10}. The results are summarized in our main Theorem \ref{toe} below, which expresses
these generating functions for any admissible initial data as explicit series with coefficients that are
Laurent polynomials with coefficients in $\Z_+[q,q^{-1}]$, 
where $q$ is the parameter of the quantum deformation.

\noindent{\bf Acknowledgments:} We thank R. Kedem for helpful discussions and the Mathematical Sciences Research
Institute, Berkeley, for hospitality during the program ``Random Matrix Theory, Interacting Particle Systems and
Integrable Systems" (fall 2010) during which this work was initiated.

\section{The quantum $A_r$ $Q$-system: definitions}

\subsection{The system}

Let $\lambda_i=i(r+1-i)$, $i=1,2,...,r$ and $q\in\C^*$. The $A_r$ quantum $Q$-system \cite{DFK10}
is an evolution equation for 
variables
$R_{i,j}$, $i\in [1,r]$ and $j\in Z$, elements of a non-commuting unital algebra:
\begin{equation}\label{clatsys}
q^{\lambda_i} \, R_{i,j+1}R_{i,j-1}=R_{i,j}^2+R_{i+1,j}R_{i-1,j}\qquad (i\in [1,r],j \in \Z)
\end{equation}
with $R_{0,j}=R_{r+1,j}=1$ for all $j\in\Z$.

\subsection{Initial data}
This is a three-term recursion relation in the variable $j$, which allows to determine 
all $R_{i,j}$ in terms of any initial data covering two consecutive values of $j$.
Initial data are indexed by Motzkin paths $\bm=(m_1,...,m_r)$
with $m_{i+1}-m_i\in \{0,1,-1\}$. They read $\bx_\bm=\big(R_{i,m_i},R_{i,m_i+1}\big)_{i\in[1,r]}$,
and are transformed into each-other via (forward/backward) mutations $\mu_i^\pm$
that act on the Motzkin paths via $\big(\mu_i^\epsilon(\bm)\big)_j=m_j+\epsilon \delta_{j,i}$,
$\epsilon=\pm$,
whenever the result is itself a Motzkin path.
The fundamental initial data corresponds to the null Motzkin path $\bm_0=(0,0,...,0)$.

Within each such set of initial data, the variables $R$ obey the following commutation relations:
\begin{equation}\label{qcom}
R_{i,j} R_{k,m}=q^{(m-j)\Lambda_{i,k}} R_{k,m}R_{i,j} 
\end{equation}
where
\begin{equation}\label{lamdef}
\Lambda_{i,k}={\rm Min}(i,k)\big(r+1-{\rm Max}(i,k)\big)
\end{equation}

\subsection{Commuting limit}

Setting $q=1$, in which case all $R_{i,j}$ variables commute, we recover the (commuting) $A_r$ $Q$-system:
\begin{equation}\label{claqsys}
R_{i,j+1}R_{i,j-1}=R_{i,j}^2+R_{i+1,j}R_{i-1,j}\qquad (i\in [1,r],j \in \Z)
\end{equation}
with $R_{0,j}=R_{r+1,j}=1$ for all $j\in\Z$.

\subsection{Quantum cluster algebra for the $A_r$ $Q$-system}

\subsubsection{Cluster algebra and quantum cluster algebra}

A cluster algebra of finite rank $n$ without coefficients is a commuting algebra generated by invertible variables
forming $n$-vectors $\bx(t)=(x_1(t),x_2(t),...,x_n(t))$ attached to the vertices $t$ of an infinite $n$-valent tree
with edges labeled $1,2,...,n$ around each vertex. Vectors $\bx(t)$ and $\bx(u)$ corresponding to vertices $t,u$
connected by an edge labeled $k$ are related by a mutation relation of the form:
\begin{eqnarray}\label{muta}
x_i(u)&=&x_i(u) \qquad {\rm for} \  i\neq k \nonumber \\
x_k(u)x_k(t)&=& \prod_{i=1}^n x_i(t)^{[B_{i,k}(t)]_+} +\prod_{i=1}^n x_i(t)^{[-B_{i,k}(t)]_+}
\end{eqnarray}
where $B(t)$ is an $n\times n$ skew-symmetrizable matrix, called the exchange matrix,
with entries in $\Z$, attached to the vertex $t$,
and subject to the mutation relation
\begin{eqnarray}\label{mutaB}
B_{i,j}(u)&=&-B_{i,j}(t) \qquad {\rm if}\ i=k\, {\rm or}\, j=k \nonumber \\
B_{i,j}(u)&=&B_{i,j}(t)+{\rm sgn}(B_{i,k}(t))\left[ B_{i,k}(t)B_{k,j}(t)\right]_+\qquad {\rm otherwise}
\end{eqnarray}
where $[x]_+={\rm Max}(x,0)$. In the following we'll be dealing only with cluster algebras with
skew-symmetric exchange matrices. In that case, we may represent each matrix $B(t)$ as
a quiver with $n$ vertices corresponding to the cluster variables, and with $B(t)_{i,j}=$
the number of arrows from vertex $i$ to vertex $j$ whenever $B(t)_{i,j}\geq 0$.

A cluster algebra is entirely specified by the pair $(\bx(t_0),B(t_0))$ of cluster variables and exchange matrix
at an initial vertex $t_0$, also called fundamental seed.

The cluster algebras have the Laurent property, that any cluster variable $x_k(t)$ may be expressed
as a Laurent polynomial of the cluster variables at any other vertex $u$ of the tree. It was conjectured in \cite{FZI}
that these polynomials have {\it non-negative} integer coefficients. 

A quantum cluster algebra \cite{BZ} of rank $n$ is a non-commuting version of the former defined as follows. 
Starting from an ordinary cluster algebra data, we introduce an extra $n\times n$ integer matrix $\Lambda(t_0)$,
forming with $B(t_0)$ a ``compatible pair", namely such that $\Lambda(t_0) B(t_0)=d$, a diagonal matrix
with positive integer entries. The matrix $\Lambda(t_0)$ encodes the quantum commutation relations
obeyed by the initial cluster variables $\bx(t_0)$, with $x_i(t_0)x_j(t_0)=q^{\Lambda_{i,j}(t_0)}x_j(t_0)x_i(t_0)$,
where $q$ is a fixed central element of the algebra.
The mutation of cluster variables is defined via an analogous, non-commuting formula, while that of
exchange matrices remains the same \eqref{mutaB}. Compatibility fixes $\Lambda(t)$ for every vertex $t$ as well.

Quantum cluster algebras also satisfy an analogous Laurent property. The positivity conjecture 
claims that the coefficients of the Laurent polynomials belong to $\Z_+[q,q^{-1}]$.

\subsubsection{Cluster algebra for the commuting $A_r$ $Q$-system}

The cluster algebra for the (commuting) $A_r$ $Q$-system \cite{Ke07}  has rank $2r$,
and a fundamental seed made of the cluster
$\bx_0\equiv \bx_{\bm_0}=(R_{1,0},R_{2,0},...,R_{r,0},R_{1,1},...,R_{r,1})$, 
and of the $2r\times 2r$ skew-symmetric exchange matrix 
$B_0\equiv B_{\bm_0}=\begin{pmatrix} 0 & -C\\ C & 0 \end{pmatrix}$, $C$ the Cartan matrix of  
$A_r$, with entries
\begin{equation}
C_{i,j}= 2\delta_{i,j}-\delta_{|j-i|,1}\qquad (i,j\in [1,r])
\end{equation}

All the initial data $\bx_\bm$ are clusters in this cluster algebra. They are obtained from $\bx_0$
via iterated cluster (forward or backward)
mutations of the form $\mu_{i}^{\pm}$ above, namely leaving
all cluster variables unchanged except  
$R_{i,m_i}\to R_{i,m_i+2}=(R_{i,m_i+1}^2+R_{i+1,m_i+1}R_{i-1,m_i+1})/R_{i,m_i}$ ($\mu_i^+$)
or $R_{i,m_i+1}\to  R_{i,m_i-1}=(R_{i,m_i}^2+R_{i+1,m_i}R_{i-1,m_i})/R_{i,m_i+1}$ ($\mu_i^-$), 
when all three terms are cluster variables in the 
original cluster. With our choice of fundamental seed, the first $r$ variables always have even indices $m$
while the $r$ next have odd ones.
The following is a sequence of forward mutations applied successively on
the fundamental initial data $x_{\bm_0}$ in the case $r=3$:
$$ {\epsfxsize=12cm \epsfbox{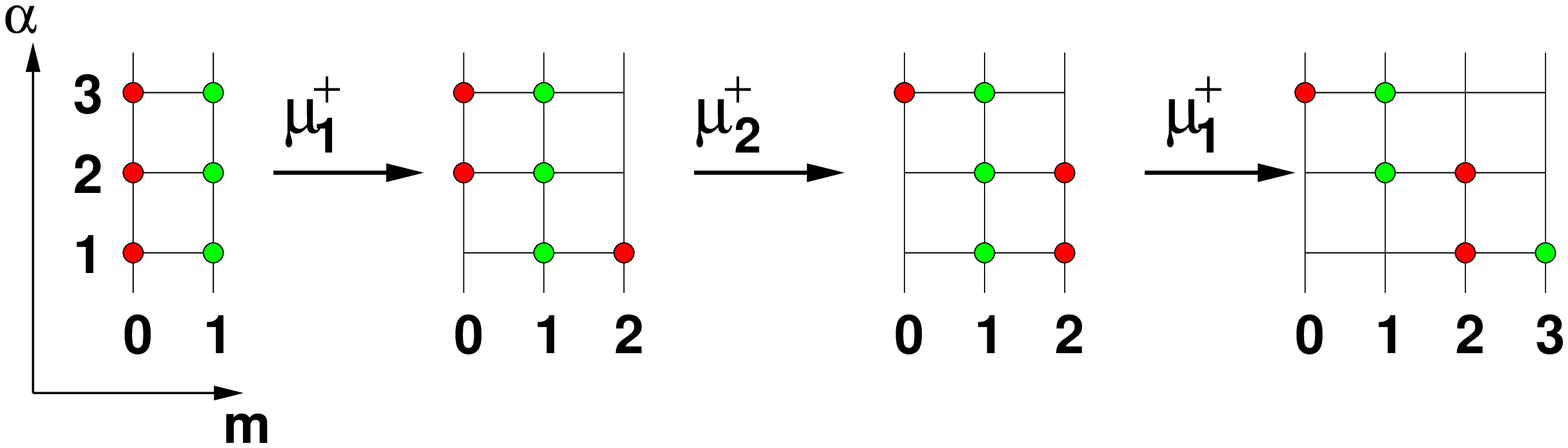}}$$
Here, each dot $(m,\al)\in \Z^2$ represents a cluster variable $R_{\al,m}$, and the corresponding
Motzkin path is the set of leftmost dots in the pairs. 

%\begin{figure}
%\centering
%\includegraphics[width=14.cm]{Bmat}
%\caption{\small }
%\label{fig:Bmat}
%\end{figure}

The (skew-symmetric) exchange matrix $B_\bm$ for each Motzkin path $\bm$ was computed explicitly in \cite{DFK08}.
As explained above, it can be represented as a quiver with $2r$ vertices
indexed by the initial data indices $(m_i,i)$ and $(m_i+1,i)$.
Here's the example for $r=3$:
$$ {\epsfxsize=12cm \epsfbox{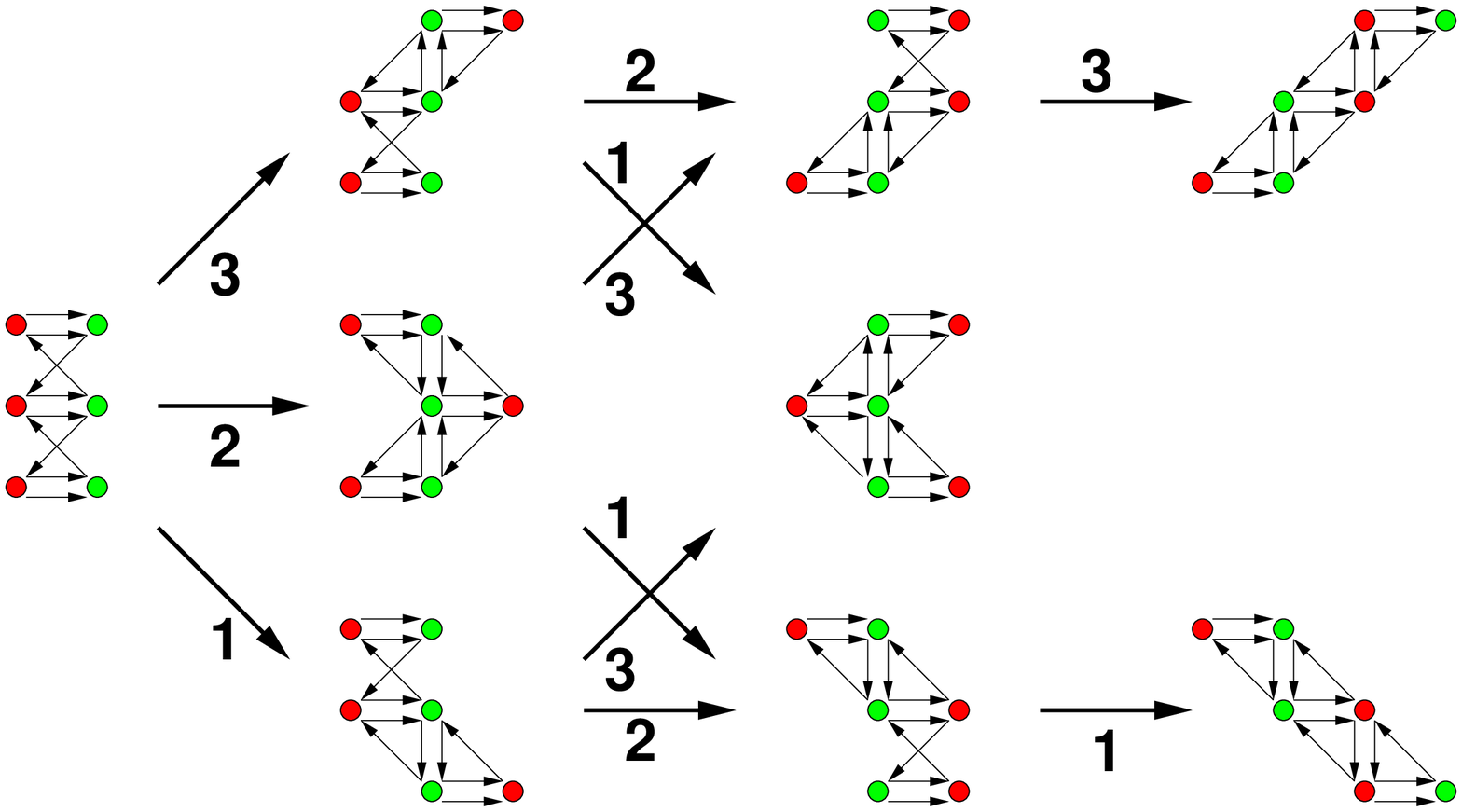}}$$
where we have represented the vertices on the same $\Z^2$ grid $(m,\al)$ as for initial data above, and where arrows
indicate various forward mutations of the corresponding index. Here, we have only represented the exchange
matrices for a fundamental set of Motzkin paths modulo a global translation by $(1,0)$, namely
the set $\cM_r=\{ (m_\al)_{\al\in[1,r]} \vert {\rm Min}_\al(m_\al)=0\}$. The exchange matrix is actually quasi-periodic,
namely: $B_{\bm+1}=-B_\bm$, where we denote by $\bm+1$ the Motzkin path $(m_1+1,m_2+1,...,m_r+1)$.

More generally, the exchange matrix $B_\bm$ is constructed as follows\footnote{This construction is due to R. Kedem.}.
Given the Motzkin path $\bm=(m_\al)$, we simply represent its vertices 
and their translates by the vector $(1,0)$ on the $(m,\al)$ plane, and we represent
either of the three following local arrow configurations, depending on whether the Motzkin path
is locally ascending ($m_{\al+1}=m_\al+1$), flat ($m_{\al+1}=m_\al$), or descending ($m_{\al+1}=m_\al-1$):
$$ {\epsfxsize=10cm \epsfbox{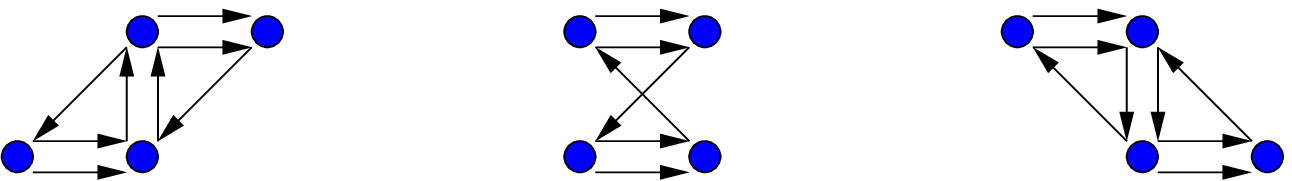}}$$
The resulting quiver encodes the skew-symmetric matrix $B_\bm$.

Finally, let us mention the following Lemma (see Ref.\cite{DFK08} for details), used  crucially in the following.

\begin{lemma}\label{lemposimut}
Any Motzkin path $\bm$ with $m_i\geq 0$ for all $i\in [1,r]$
may be attained from $\bm_0$
via iteration of forward mutations of the form $\mu_i^+$ acting at each 
intermediate step on a Motzkin path $\bm$
in either of the two following local configurations around $i$: 
\begin{itemize}
\item Case (i): $m_{i-1}=m_i=m_{i+1}-1$
\item Case (ii): $m_{i-1}=m_i=m_{i+1}$
\end{itemize}
\end{lemma}

\subsubsection{Quantum cluster algebra for the quantum $A_r$ $Q$-system}

The quantum cluster algebra corresponding to our quantum $Q$-system \eqref{clatsys}
has the fundamental seed
$\bx_0=((R_{i,0})_{i=1}^r,(R_{i,1})_{i=1}^r)$, and the same exchange matrix $B_0$ as in the commuting
case. The commutation relations \eqref{qcom} correspond to taking the initial compatible
pair $(\Lambda_0,B_0)$ such that $\Lambda_0=(r+1)B_0^{-1}$. The compatibility implies
that $\Lambda_\bm=(r+1)B_\bm^{-1}$ for all Motzkin paths $\bm$, and we have
explicitly
$ (\Lambda_\bm)_{(i,m),(j,p)}=(p-m)\Lambda_{i,j} $ in terms of the matrix $\Lambda$ of eq.\eqref{lamdef},
for all pairs $(i,m)$ and $(j,p)$ of cluster indices in $\bx_\bm$, which leads to the commutations \eqref{qcom}.

\section{Quantum system solution for $R_{1,n}$ via continued fractions}

\subsection{Generating functions}

We set $R_n\equiv R_{1,n}$.
To each Motzkin path $\bm=(m_1,m_2,...,m_r)$, we associate the generating function 
\begin{equation}\label{Fmdef}
F_\bm(t)=\sum_{n=0}^\infty t^n R_{n+m_1}R_{m_1}^{-1} 
\end{equation}
We also use the ``rerooted" generating function
\begin{equation}\label{Gmdef}
G_\bm(t)=\sum_{n=0}^\infty t^n R_{n+m_1+1}R_{m_1+1}^{-1} 
\end{equation}

\subsection{Continued fraction expressions}

\begin{figure}
\centering
\includegraphics[width=10.cm]{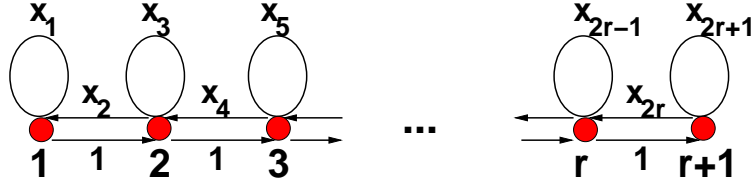}
\caption{\small The weighted graph $\Gamma_r$, with $r+1$ vertices labeled $1,2,...,r+1$.
For each oriented edge $e$, we have indicated the corresponding step weight $w(e)$, namely: $w(i\to i)=x_{2i-1}$
($i=1,2,...,r+1$),
$w(i\to i+1)=x_{2i}$, $w(i+1\to i)=1$ ($i=1,2,...,r$).}
\label{fig:gammar}
\end{figure}

For some variables $x_1,x_2,...,x_{2r+1}$ elements of a non-commuting algebra,
let $J_1(x_1,x_2,...,x_{2r+1})$ be the ``non-commutative Jacobi-type (finite) continued fraction" defined inductively by
\begin{eqnarray*}
J_s(x_{2s-1},x_{2s},...,x_{2r+1})&=&\big(1-x_{2s-1}-J_{s+1}(x_{2s+1},x_{2s+2},...,x_{2r+1})x_{2s}\big)^{-1}
\quad (1\leq s\leq r+1)\\
J_{r+2}&=&0
\end{eqnarray*}
For some central scalar parameter $t$,
the function $J_1(tx_1,...,tx_{2r+1})$, expanded as a formal power series of $t$, 
may be interpreted combinatorially as the generating series 
for ``quantum" paths on the weighted graph $\Gamma_r$ depicted in Fig.\ref{fig:gammar}, 
from and to the origin vertex $1$,
each path being weighted by the product of its step weights taken in the order they are traversed
and with an extra weight $t$ per step along each edge pointing toward the origin.

\begin{defn}
To each Motzkin path $\bm$ we attach a sequence of ``weights"
$\by(\bm)=(y_1(\bm),...,y_{2r+1}(\bm))$ via the following induction
under forward mutation $\mu_i^+(\bm)=\bm'$, depending on whether $\bm$ is
in cases (i) or (ii) of Lemma \ref{lemposimut} above. For short we write $y_i(\bm)=y_i$ and $y_i(\bm')=y_i'$.
First, we have $y_j'=y_j$ for $j\neq 2i-1,2i,2i+1$ (case (i)) and
$j\neq 2i-1,2i,2i+1,2i+2$ (case (ii)), while:
\begin{eqnarray}
&&{\rm \bf Cases\  (i)\  and\  (ii):} \left\{ \begin{matrix} y_{2i-1}'=y_{2i-1}+y_{2i}\hfill \\
\quad y_{2i}'=y_{2i+1}y_{2i}(y_{2i-1}+y_{2i})^{-1}\hfill \nonumber \\
y_{2i+1}'=y_{2i+1}y_{2i-1}(y_{2i-1}+y_{2i})^{-1}\hfill
\end{matrix} \right. \nonumber \\
&& \qquad \qquad \ \  {\rm \bf Case (ii):}\ \ \  y_{2i+2}'=y_{2i+1}y_{2i-1}(y_{2i-1}+y_{2i})^{-1} \label{recuy}
\end{eqnarray}
\end{defn}

This determines the $y$'s entirely in terms of the initial data $\by(\bm_0)$. We have the following
Theorems.

\begin{thm}\label{ronesol}(\cite{DFK10})
For the fundamental initial data with $\bm=\bm_0=(0,0,...,0)$,
the solution to the quantum $A_r$ $Q$-system satisfies the following identity:
\begin{equation}\label{fraczero}
F_{\bm_0}(t)=1+tG_{\bm_0}(t)y_1(\bm_0), \quad G_{\bm_0}(t)=J_1(ty_1(\bm_0),...,ty_{2r+1}(\bm_0))
\end{equation}
where 
\begin{equation}\label{yzero}
y_{2i-1}(\bm_0)=R_{i,1}R_{i-1,1}^{-1}R_{i,0}^{-1} R_{i-1,0}
\qquad y_{2i}(\bm_0)=R_{i+1,1}R_{i,1}^{-1}R_{i,0}^{-1}R_{i-1,0} 
\end{equation}
\end{thm}

\begin{thm}\label{ysol}(\cite{DFK10})
For $\by(\bm_0)=(y_1(\bm_0),...,y_{2r+1}(\bm_0))$ given by \eqref{yzero}, the solution of the 
recursion \eqref{recuy} reads:
\begin{eqnarray}
\qquad y_{2i-1}(\bm)&=&
q^{i-1} \, R_{i,m_i+1}R_{i,m_i}^{-1} 
 R_{i-1,m_{i-1}}R_{i-1,m_{i-1}+1}^{-1}
\label{oddyq} \\
y_{2i}(\bm)&=&\left\{  \begin{matrix} 
R_{i+1,m_{i+1}+1}R_{i+1,m_{i+1}}^{-1}
R_{i+1,m_{i}}R_{i+1,m_{i}+1}^{-1}
& {\rm if} \ m_i=m_{i+1}+1\\
{1} & {\rm otherwise} 
\end{matrix}\right\}\nonumber \\
&&\times\, 
 R_{i+1,m_i+1}R_{i,m_i+1}^{-1}
 R_{i,m_i}^{-1}R_{i-1,m_i} \nonumber \\
&& \times \left\{ \begin{matrix} 
R_{i-1,m_{i}+1}R_{i-1,m_{i}}^{-1}
R_{i-1,m_{i-1}}R_{i-1,m_{i-1}+1}^{-1}
& {\rm if} \ m_i=m_{i-1}-1\\
{1} & {\rm otherwise} 
\end{matrix}\right\}\label{evenyq} 
\end{eqnarray}
\end{thm}

\begin{thm}\label{yzer}(\cite{DFK10})
For the $y_i(\bm)$ as in Theorem \ref{ysol},
we have:
$$F_\bm(t)=1+t G_\bm(t)y_1(\bm), \quad  G_\bm(t)=J_1(\hy_1(\bm),...,\hy_{2r+1}(\bm)) $$
where $\hy_i\equiv \hy_i(\bm)$ are related to $y_i\equiv y_i(\bm)$ via:
\begin{equation}\left\{
\begin{matrix}
\hy_{2i-1}=t y_{2i-1}\hfill & \hy_{2i}=t y_{2i} \hfill & {\rm if}\ m_{i+1}=m_i \hfill \\
\hy_{2i-1}=t (y_{2i-1}+y_{2i})\hfill & \hy_{2i}=t^2 y_{2i+1}y_{2i}\hfill & {\rm if}\ m_{i+1}=m_i+1 \hfill \\
\hy_{2i-1}=t y_{2i-1}-y_{2i+1}^{-1}y_{2i}\hfill & \hy_{2i}=y_{2i+1}^{-1}y_{2i}
\hfill & {\rm if}\ m_{i+1}=m_i-1 \hfill 
\end{matrix} \right .\qquad (i\in [1,2r+1])
\end{equation}
\end{thm}

\begin{example}\label{exzero}
For $\bm=\bm_0=(0,0,...,0)$, we have $\hy_i=t y_i$ for $i=1,2,...,2r+1$. The generating function
reads
$$G_{\bm_0}(t)=\big(1-ty_1-t\big(1-t y_3-t(... (1-ty_{2r-1}-t(1-ty_{2r+1})^{-1}y_{2r})^{-1} ...)^{-1}y_4)^{-1}y_2\big)^{-1}$$
\end{example}
\begin{example}\label{exstiel}
For $\bm=\bm_1=(0,1,...,r-1)$, we have $\hy_{2i-1}=t (y_{2i-1}+y_{2i})$ and $\hy_{2i}=t^2y_{2i+1}y_{2i}$
for all $i$, where $y_i\equiv y_i(\bm_1)$. The result is simplest for $F_{\bm_1}(t)=1+tG_{\bm_1}(t)y_1$.
The generating function may be rearranged using the following identity at each step:
\begin{equation}\label{reart} a+b +(1-c-u)^{-1} cb =a+(1-c-u)^{-1}(1-u)b=a+(1-(1-u)^{-1}c)^{-1}b\end{equation}
The result reads
$$F_{\bm_1}(t)=\big(1-t\big(1-t (... (1-t(1-ty_{2r+1})^{-1}y_{2r})^{-1} ...)^{-1}y_2)^{-1}y_1\big)^{-1}$$
This is easily expressed in terms of the ``non-commutative Stieltjes-type (finite) continued fraction"
defined inductively as
$$ S_k(x_k,x_{k+1},...,x_{2r+1})=(1-S_{k+1}(x_{k+1},...,x_{2r+1})x_k)^{-1} \ (1\leq k\leq 2r+1), \ S_{2r+2}=1$$
via:
$ F_{\bm_1}(t)=S_1(ty_1,ty_2,...,ty_{2r+1})$.
\end{example}

In view of the Example \ref{exstiel}, we may write another (mixed Stieltjes-Jacobi-type) continued
fraction expression for $G_\bm(t)$ for arbitrary $\bm$. This will be crucially used in the following.

Any Motzkin path $\bm=(m_1,...,m_r)$ may be decomposed into strictly ascending segments
of the form $(m,m+1,...,m+k-1)$ separated by weakly descending steps of the form $(m,m)$ or $(m,m-1)$.
Accordingly, we may transform the Jacobi-type fraction for $G_\bm(t)$ in Theorem \ref{yzer}, by
``undoing" the pieces of the continued fraction that correspond to the strictly ascending
segments of $\bm$. 

To this end, assume that we have a strictly ascending
segment $(m_{i_0},m_{i_0+1},...,m_{i_0+k-1})$ with $m_{i_0+j}=m_{i_0}+j$ for $j=0,1,...,k-1$, which is followed
by a weakly decreasing step $(m_{i_0+k-1},m_{i_0+k})$, with $m_{i_0+k}=m_{i_0+k-1}$
or $m_{i_0+k}=m_{i_0+k-1}-1$.
Then by definition, we have the relations
\begin{eqnarray}
J_{i_0+j} &=& \big(1-t(y_{2i_0+2j-1}+y_{2i_0+2j})-t^2 J_{i_0+j+1} y_{2i_0+2j+1}y_{2i_0+2j}\big)^{-1}
\quad (j\in[0,k-2])\nonumber \\
\quad \qquad J_{i_0+k-1}&=& \big(1-ty_{2i_0+2k-3}-(J_{i_0+k}-\delta_{m_{i_0+k},m_{i_0+k-1}-1}) \hy_{2i_0+2k-2}\big)^{-1}
\label{initJ}
\end{eqnarray}
Like in Example \ref{exzero}, let us rearrange this piece of fraction iteratively from the bottom up,
by using the relation \eqref{reart}. This allows to rewrite:
\begin{eqnarray}
J_{i_0}&=&(1-t y_{2i_0-1}- t S_{2i_0}y_{2i_0} )^{-1} \nonumber \\
S_{2i_0+j}&=&\big(1-t S_{2i_0+j+1} y_{2i_0+1+j}\big)^{-1} \qquad (j\in[0,2k-4]) \nonumber \\
S_{2i_0+2k-3}&=&\big(1-S_{2i_0+2k-2} \hy_{2i_0+2k-2}\big)^{-1} \nonumber \\
S_{2i_0+2k-2}&=&J_{i_0+k}-\delta_{m_{i_0+k},m_{i_0+k-1}-1}\label{Spart}\end{eqnarray}
Indeed, we start by using \eqref{reart} at the step $j=k-2$ of \eqref{initJ}, with $a=ty_{2i_0+2k-5}$, $b=ty_{2i_0+2k-4}$, 
$c=y_{2i_0+2k-3}$ and $u=S_{2i_0+2k-2} \hy_{2i_0+2k-2}$, to rewrite:
$$
J_{i_0+k-2}=\Big(1- ty_{2i_0+2k-5} -t\big(1-t(1-S_{2i_0+2k-2}
 \hy_{2i_0+2k-2})^{-1}y_{2i_0+2k-3}\big)^{-1}y_{2i_0+2k-4}\Big)^{-1}
$$
Iterating this leads straightforwardly to \eqref{Spart}.
Repeating this transformation for every strictly ascending segment of $\bm$, we arrive at:

\begin{thm}\label{mixSJ}
For any Motzkin path $\bm$,
we have the following mixed Stieltjes-Jacobi type continued fraction expression for $G_{\bm}$. Let $i_1,..,i_s$
and $\ell_1,...,\ell_s$ be the positions and lengths of the strictly ascending segments of $\bm$,
of the form $(m_{i_a},m_{i_a}+1,...,m_{i_a+\ell_a-1})$, $a=1,2,...,s$. Then we have:
\begin{eqnarray*}
G_\bm(t)=& J_1   & \\
J_i =&\big(1- \hy_{2i-1} -J_{i+1} \hy_{2i}\big)^{-1} \hfill &(i\not\in \cup_a [i_a,i_a+\ell_a-1])\\
J_{i_a}=&\big(1- ty_{2i_a-1} -S_{2i_a} ty_{2i_a}\big)^{-1}  \hfill &(a \in [1,s])\\
S_i =&\big(1 -S_{i+1} ty_{i+1}\big)^{-1} \hfill   &(i\in [2i_a,2i_a+2\ell_a-4];a\in[1,s])\\
S_{2i_a+2\ell_a-3}=&\big(1 -S_{2i_a+2\ell_a-2} \hy_{2i_a+2\ell_a-2}\big)^{-1} \hfill  &(a\in[1,s])\\
S_{2i_a+2\ell_a-2}=&J_{i_a+\ell_a}-\delta_{m_{i_a+\ell_a},m_{i_a+\ell_a-1}-1} \hfill  &(a\in[1,s])
\end{eqnarray*}
\end{thm}

\subsection{Quantum commutation relations for the weights}

Using the commutations \eqref{qcom}, we obtain:
\begin{thm}\label{mzer}
Introducing $p=q^{r+1}$,
the weights $y_i\equiv y_i(\bm_0)$ of Theorem \ref{ronesol}
obey the following $p$-commutation relations:
\begin{eqnarray*}
y_{i}y_{i+1}&=&p \, y_{i+1}y_{i}\qquad (i\in [1,2r])\\
y_{2i}y_{2i+2}&=& p \, y_{2i+2}y_{2i} \qquad (i\in [1,r-1])\\
y_j y_k &=& y_k y_j \qquad {\rm otherwise}
\end{eqnarray*}
\end{thm}

Using the recursion relations \eqref{recuy}, we deduce the commutation
relations of $y_i(\bm)$, for arbitrary Motzkin paths $\bm$:

\begin{thm}\label{ycom}
For a given Motzkin path $\bm$, the weights $y_i\equiv y_i(\bm)$ have the 
following commutation relations
\begin{eqnarray}
y_{i}y_{i+1}&=&p \, y_{i+1}y_{i}\qquad (i\in [1,2r])\label{commyone}\\
y_{2i}y_{2i+2}&=&p^{m_{i}-m_{i+1}+1}\,y_{2i+2} y_{2i} \qquad (i\in [1,r-1])\label{commytwo}\\
y_j y_k &=& y_k y_j \qquad {\rm otherwise}\nonumber 
\end{eqnarray}
\end{thm}
\begin{proof}
By induction under mutation. The Theorem holds for $\bm=\bm_0$ (Theorem \ref{mzer}).
Assume it holds for some $\bm$, then consider $\bm'=\mu_i^+(\bm)$ in either cases (i) or (ii)
of Lemma \ref{lemposimut},
and denote by $y_j'\equiv y_j(\bm')$.
We deduce  the following commutations from the recursion hypothesis (eq.\eqref{commyone}):
\begin{eqnarray}
y_{2j-1}y_{2k-1}&=&y_{2k-1}y_{2j-1}\quad (j,k\in [1,r+1])\label{one}\\
(y_{2j-1}+y_{2j}) y_{2j+1}y_{2j}&=&p \, y_{2j+1}y_{2j}(y_{2j-1}+y_{2j}) \qquad (j\in [1,r])\label{two}\\
(y_{2j-1}+y_{2j}) y_{2j+1}y_{2j-1}&=&y_{2j+1}y_{2j-1} (y_{2j-1}+y_{2j}) \qquad (j\in [1,r])\label{three}
\end{eqnarray}
In both cases (i) and (ii), this implies that $y_{2j-1}'y_{2k-1}'=y_{2k-1}'y_{2j-1}'$ for all $j,k$
and that, as $m_{i-1}=m_i$ in both cases, $y_{2i-2}'y_{2i-1}'=p\, y_{2i-1}'y_{2i-2}'$, while
$y_{2i-1}' y_{2i}'=p\, y_{2i}'y_{2i-1}'$ as a consequence of \eqref{two} and
$y_{2i}'y_{2i+1}'=p\, y_{2i+1}'y_{2i}'$ by use of \eqref{two}-\eqref{three}.
Finally, we get $y_{2i-2}'y_{2i}'=y_{2i}'y_{2i-2}'$ as $y_{2i-2}'=y_{2i-2}$ $p$-commutes with 
both $y_{2i-1}+y_{2i}$ and  $y_{2i}$.
Using the expressions \eqref{recuy} for the cases (i) and (ii), we finally find
\begin{itemize}
\item Case (i): $y_{2i}'y_{2i+2}'=p\, y_{2i+2}'y_{2i}'$, as $y_{2i+2}'=y_{2i+2}$
and  $y_{2i}y_{2i+2}=y_{2i+2}y_{2i}$
(from $m_{i+1}=m_i+1$).
\item Case (ii): $y_{2i}'y_{2i+2}'=p^2\, y_{2i+2}'y_{2i}'$, as $(y_{2i-1}+y_{2i})$ commutes
with $y_{2i+2}y_{2i-1}$ (due to $y_{2i}y_{2i+2}=p\, y_{2i+2}y_{2i}$, from $m_{i+1}=m_i$).
\end{itemize}
Noting finally that $m_{i-1}'-m_i'+1=0$ in both cases (i)-(ii),
and that $m_i'-m_{i+1}'+1=1$ in case (i) and $=2$ in case (ii), the Theorem follows.
\end{proof}

\begin{figure}
\centering
\includegraphics[width=12.cm]{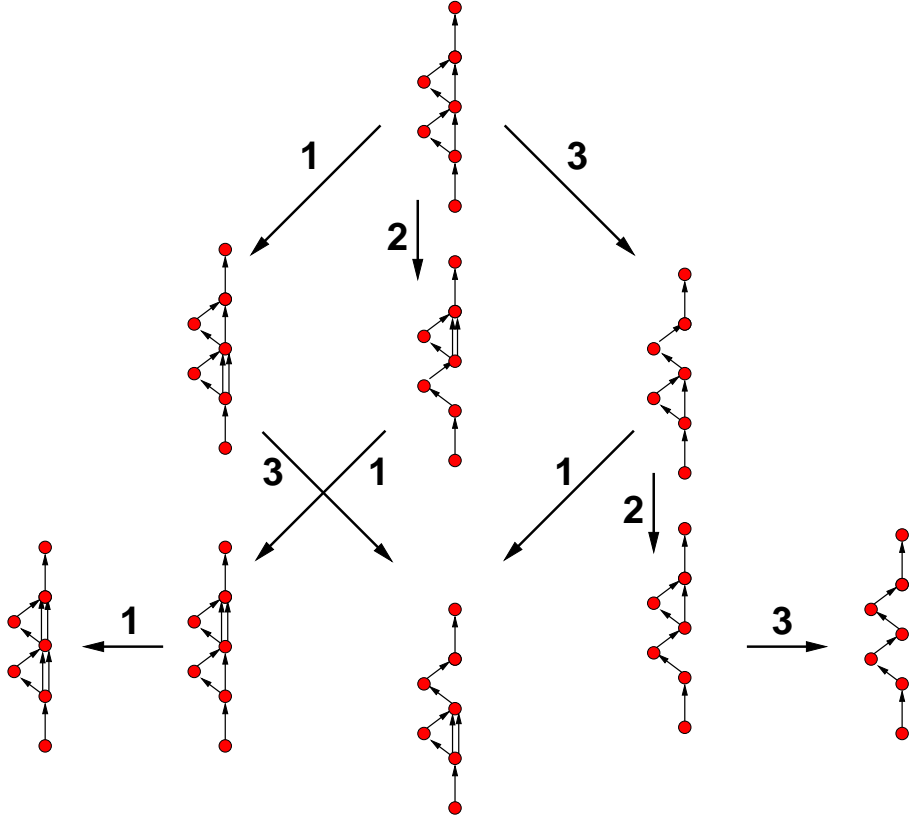}
\caption{\small The nine quivers $Q_\bm$ corresponding to the nine
Motzkin paths $\bm=(m_1,m_2,m_3)\in \cM_3$, and $Q_{\bm_0=(0,0,0)}$
on top of the diagram. In each quiver, the vertices are all labeled from bottom to top
$1,2,...,7$.
The arrows between quivers, labeled $i$, correspond to mutations $\mu_i^+$ in either
case (i) or (ii).}
\label{fig:athreecomm}
\end{figure}

To each Motzkin path $\bm$ we may associate a quiver $Q_\bm$ with $2r+1$ vertices
labelled $i=1,2,...,2r+1$, that summarizes
the commutations of the $y_i(\bm)$'s as follows: we draw $m$ arrows from
vertex $i$ to vertex $j$ whenever $y_i y_j=p^m\, y_j y_i$.  For illustration of Theorem \ref{ycom},
we have depicted in Fig.\ref{fig:athreecomm} the example $r=3$ of
the quivers $Q_\bm$ in the fundamental domain $\cM_3$
of Motzkin paths under global integer translations, 
and indicated by arrows and superscripts $i$
the mutations $\mu_i^+$ acting on them.

\subsection{Quantum multinomial expressions}

\subsubsection{$p$-combinatorics}

For $a_1,...,a_k\in \Z_+$, the quantum multinomial coefficient is defined as:
$$
\left[\begin{matrix} a_1+\cdots+ a_k \\ a_1,\ldots, a_k\end{matrix}\right]_p=
{\prod_{i=1}^{a_1+\cdots +a_k} (1-p^i) \over \prod_{j=1}^k\prod_{i=1}^{a_j} (1-p^i)}
$$
We assume by convention that when $k=2$, $\left[\begin{matrix}-1\\ 0\end{matrix}\right]_p=1$.
We have the following
$p$-multinomial identity for $k$ variables $x_1,...,x_k$ such that $x_ix_j=px_j x_i$ for all $1\leq i<j\leq k$:
$$(x_1+\cdots+x_k)^n=\sum_{m_1,...,,m_k \geq 0\atop m_1+\cdots m_k= n} 
\left[\begin{matrix} n \\ m_1,\ldots, m_k\end{matrix}\right]_p x_k^{m_k} x_{k-1}^{m_{k-1}}\cdots x_1^{m_1}$$
This reduces to the standard $p$-binomial identity for $k=2$.

We also define the following formal generating series:
$$\phi_\ell(z)=\prod_{i=0}^{\ell-1} (1-p^i z)^{-1} =\sum_{m=0}^\infty \left[\begin{matrix} \ell+m-1 \\ m\end{matrix}\right]_p z^m$$
and, for variables $x_1,...,x_k$ such that $x_ix_j=px_j x_i$ for all $1\leq i<j\leq k$, we have:
$$\phi_\ell(x_1+\cdots +x_k)
=\sum_{m_1,\ldots m_k\geq 0}   \left[\begin{matrix} \ell-1+m_1+\cdots+m_k \\ \ell-1,m_1,\ldots, m_k\end{matrix}\right]_p 
x_k^{m_k} x_{k-1}^{m_{k-1}}\cdots x_1^{m_1}$$

\subsubsection{Explicit expressions for $G_\bm(t)$}
We have the following explicit expression for $G_{\bm_0}(t)$, after ordering of the $y_i=y_i(\bm_0)$'s.
\begin{thm}
For the flat Motzkin path $\bm_0=(0,0,...,0)$, we have:
\begin{eqnarray*}
G_{\bm_0}(t)&=&\sum_{\ell_1,\ell_2,...,\ell_{2r+1}\geq 0\atop
\ell_0=1,\ell_{2r+2}=0} \prod_{i=0}^r \left[\begin{matrix}\ell_{2i}+\ell_{2i+1}+\ell_{2i+2}-1\\ 
\ell_{2i}-1,\ell_{2i+1},\ell_{2i+2}\end{matrix}\right]_p (ty_{2r+1})^{\ell_{2r+1}}(ty_{2r})^{\ell_{2r}}\cdots (ty_1)^{\ell_1}\\
F_{\bm_0}(t)&=&1+\sum_{\ell_1\geq 1,\ell_2,...,\ell_{2r+1}\geq 0\atop
\ell_0=1,\ell_{2r+2}=0} \prod_{i=0}^r \left[\begin{matrix}\ell_{2i}+\ell_{2i+1}+\ell_{2i+2}-1\\ 
\ell_{2i}-1,\ell_{2i+1},\ell_{2i+2}\end{matrix}\right]_p (ty_{2r+1})^{\ell_{2r+1}}(ty_{2r})^{\ell_{2r}}\cdots (ty_1)^{\ell_1}
\end{eqnarray*}
\end{thm}
\begin{proof}
By induction. We start with the expression \eqref{fraczero} for $G_{\bm_0}(t)=J_1$,
and write the inductive definition $J_1=(1-ty_1-tJ_2 y_2)^{-1}$. Next we note that $y_1 (J_2y_2)=p(J_2y_2)y_1$,
as $J_2$ only involves $y_i$, $i\geq 3$, which all commute with $y_1$, by Theorem \ref{mzer}. We deduce the
formal expansion 
$$J_1=\sum_{\ell_1,\ell_2\geq 0} \left[\begin{matrix}\ell_{1}+\ell_{2}\\ 
\ell_{2}\end{matrix}\right]_p (tJ_2y_2)^{\ell_2} (ty_1)^{\ell_1}$$
Assume we have
\begin{equation}
\label{hyprec}
J_1=\sum_{\ell_0,\ell_1,...,\ell_{2m}\geq 0\atop \ell_0=1} \prod_{i=0}^{m-1}\left[\begin{matrix}\ell_{2i}+\ell_{2i+1}+\ell_{2i+2}-1\\ 
\ell_{2i}-1,\ell_{2i+1},\ell_{2i+2}\end{matrix}\right]_p 
(tJ_{m+1}y_{2m})^{\ell_{2m}} (ty_{2m-1})^{\ell_{2m-1}}\cdots (ty_1)^{\ell_1}
\end{equation}
for some $m\geq 1$, then by Theorem \ref{mzer}, we know that $y_{2m}y_{2m+1}=py_{2m+1}y_{2m}$, 
$y_{2m}y_{2m+2}=py_{2m+2}y_{2m}$  while $y_{2m}$
commutes with $y_i$, $i\geq 2m+3$, and in particular with $J_{m+2}$. This implies:
\begin{eqnarray*}
y_{2m}J_{m+1}&=&y_{2m}(1-ty_{2m+1}-tJ_{m+2}y_{2m+2})^{-1}=(1-pty_{2m+1}-ptJ_{m+2}y_{2m+2})^{-1}y_{2m}\\
(tJ_{m+1}y_{2m})^\ell&=&\phi_\ell(ty_{2m+1}+tJ_{m+2}y_{2m+2})(ty_{2m})^{\ell}\\
&=&\sum_{\ell_{2m+1},\ell_{2m+2}\geq 0}
\left[\begin{matrix}\ell+\ell_{2m+1}+\ell_{2m+2}-1\\ 
\ell-1,\ell_{2m+1},\ell_{2m+2}\end{matrix}\right]_p (tJ_{m+2}y_{2m+2})^{\ell_{2m+2}}(ty_{2m+1})^{\ell_{2m+1}}(ty_{2m})^{\ell}
\end{eqnarray*}
where in the last line we have used the fact that $y_{2m+1}y_{2m+2}=py_{2m+2}y_{2m+1}$,
while $y_{2m+1}$ commutes with $J_{m+2}$. Substituting this into \eqref{hyprec}, with $\ell=\ell_{2m}$,
we obtain the same summation with $m\to m+1$. We conclude that \eqref{hyprec} holds for all $m$,
and in particular for $m=r+1$, in which case $J_{r+2}=0$ imposes that the last summation reduce to $\ell_{2r+2}=0$,
and the first part of the Theorem follows. The second follows trivially from the relation between $F_{\bm_0}$
and $G_{\bm_0}$.
\end{proof}

\begin{thm} \label{monethm}
For the ascending Motzkin path $\bm_1=(0,1,2,...,r-1)$, we have:
\begin{eqnarray*}
G_{\bm_1}(t)&=&\sum_{\ell_1,\ell_2,...,\ell_{2r+1}\geq 0} \prod_{i=1}^{2r} \left[\begin{matrix}\ell_{i}+\ell_{i+1}-1+\delta_{i,1}\\ 
\ell_{i+1}\end{matrix}\right]_p (ty_{2r+1})^{\ell_{2r+1}}(ty_{2r})^{\ell_{2r}}\cdots (ty_1)^{\ell_1}\\
F_{\bm_1}(t)&=&\sum_{\ell_1,\ell_2,...,\ell_{2r+1}\geq 0} 
\prod_{i=1}^{2r} \left[\begin{matrix}\ell_{i}+\ell_{i+1}-1\\ 
\ell_{i+1}\end{matrix}\right]_p (ty_{2r+1})^{\ell_{2r+1}}(ty_{2r})^{\ell_{2r}}\cdots (ty_1)^{\ell_1}
\end{eqnarray*}
\end{thm}
\begin{proof}
We start from the expression $F_{\bm_1}(t)=S_1$ of Example \ref{exstiel}, and write
$$S_1=(1-tS_2y_1)^{-1}=\sum_{\ell_1\geq 0} \phi_{\ell_1}(tS_3y_2)(ty_1)^{\ell_1}=\sum_{\ell_1,\ell_2\geq 0}
 \left[\begin{matrix}\ell_{1}+\ell_{2}-1\\ 
\ell_{2}\end{matrix}\right]_p (tS_3^{\ell_2}y_2)(ty_1)^{\ell_1}$$
Assume we have
$$S_1=\sum_{\ell_1,...,\ell_{m-1}\geq 0} \prod_{i=1}^{m-2} \left[\begin{matrix}\ell_{i}+\ell_{i+1}-1\\ 
\ell_{i+1}\end{matrix}\right]_p (tS_my_{m-1})^{\ell_{m-1}}(ty_{m-2})^{\ell_{m-2}} \cdots (ty_1)^{\ell_1}$$
for some $m$, then writing $(tS_my_{m-1})^{\ell}=\phi_\ell(tS_{m+1}y_m)y_{m-1}^\ell$ immediately
implies the same relation for $m\to m+1$. We conclude that it holds for all $m$, in particular for $m=r+2$,
where it boils down to the second part of the Theorem. The first part follows from the relation between 
$F_{\bm_1}$ and $G_{\bm_1}$.
\end{proof}

\subsubsection{The main Theorem}

More generally, we have

\begin{thm}\label{toe}
For a generic Motzkin path $\bm$, we have the following:
\begin{eqnarray*}
F_{\bm}(t)&=&\sum_{\ell_1,\ell_2,...,\ell_{2r+1}\in \Z_+} A_\bm(\ell_1,\ell_2,...,\ell_{2r+1}) (ty_{2r+1})^{\ell_{2r+1}}\cdots
(ty_1)^{\ell_1}\\
G_{\bm}(t)&=&\sum_{\ell_1,\ell_2,...,\ell_{2r+1}\in\Z_+} A_\bm(\ell_1+1,\ell_2,...,\ell_{2r+1}) (ty_{2r+1})^{\ell_{2r+1}}\cdots
(ty_1)^{\ell_1}
\end{eqnarray*}
where $A_\bm$ is defined as the product 
\begin{eqnarray*}
A_\bm(\ell_1,...,\ell_{2k+1})&=&\left[\begin{matrix}\ell_{1}+\ell_{2}-1\\ 
\ell_{2}\end{matrix}\right]_p \qquad \raisebox{-.5cm}{\hbox{\epsfxsize=.5cm \epsfbox{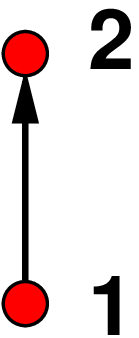}}} \\
&&\!\!\!\!\!\!\!\!\!\!\!\!\!\!\!\!\!\!\!\!\!\!\!\!\!\!\!\!\!\!\!\!\!\!\!\!\!\!\!\!\!\!\!\!\!\!\!\!\!\!\!\!\!\!\!\!
\times \prod_{i=1}^{r-2} \left\{ \begin{matrix}
\left[\begin{matrix}\ell_{2i}+\ell_{2i+1}+\ell_{2i+2}-1\\ 
\ell_{2i}-1,\ell_{2i+1},\ell_{2i+2}\end{matrix}\right]_p & {\rm if}\, m_{i+1}=m_i 
&\raisebox{-1.cm}{\hbox{\epsfxsize=2.5cm \epsfbox{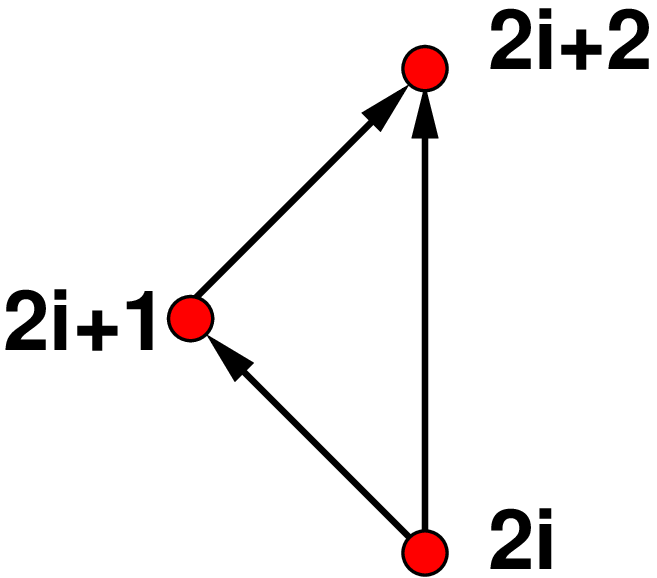}}}\\
\left[\begin{matrix}\ell_{2i}+\ell_{2i+1}+\ell_{2i+2}-1\\ 
\ell_{2i}-1\end{matrix}\right]_p
\left[\begin{matrix}\ell_{2i}+\ell_{2i+1}+\ell_{2i+2}\\ 
\ell_{2i+2}\end{matrix}\right]_p & {\rm if}\, m_{i+1}=m_i-1
&\raisebox{-1.cm}{\hbox{\epsfxsize=2.5cm \epsfbox{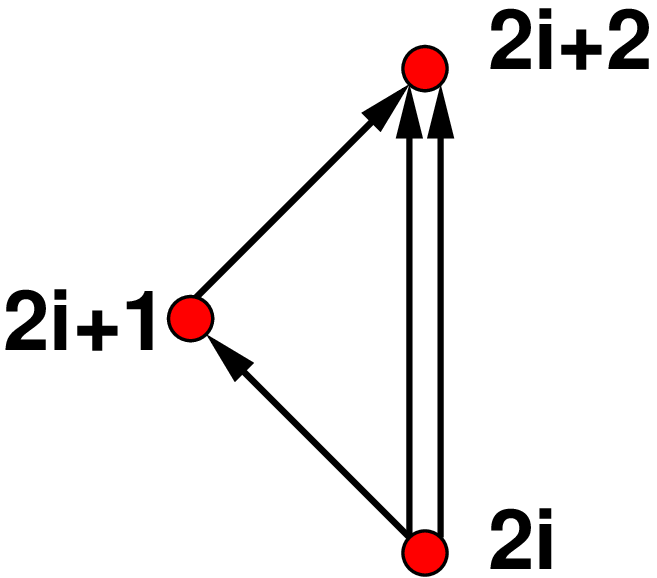}}}\\
\left[\begin{matrix}\ell_{2i}+\ell_{2i+1}-1\\ 
\ell_{2i+1}\end{matrix}\right]_p
\left[\begin{matrix}\ell_{2i+1}+\ell_{2i+2}-1\\ 
\ell_{2i+2}\end{matrix}\right]_p & {\rm if}\, m_{i+1}=m_i+1
&\raisebox{-1.cm}{\hbox{\epsfxsize=2.5cm \epsfbox{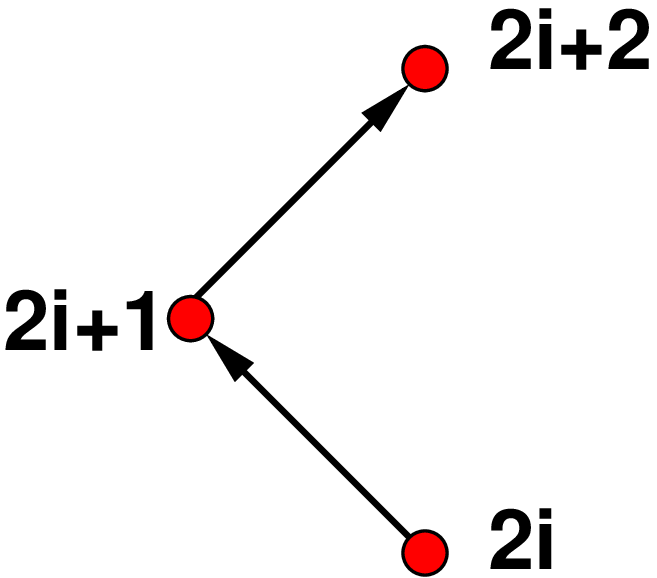}}}
\end{matrix} \right\}\\
&&\!\!\!\!\times \left[\begin{matrix}\ell_{2r}+\ell_{2r+1}-1\\ 
\ell_{2r+1}\end{matrix}\right]_p 
\qquad \raisebox{-.5cm}{\hbox{\epsfxsize=1.cm \epsfbox{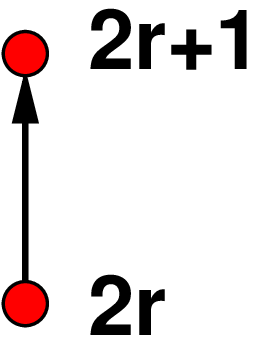}}} 
\end{eqnarray*}
where we have represented the local structure of the corresponding quiver $Q_\bm$
that encodes the $p$-commutations of the $y$'s.
\end{thm}

\begin{proof}
As before, we proceed by descending induction. Assume that for some $k\geq 1$, such that 
$m_k=m_{k-1}$ (Case (1)) or $m_k=m_{k-1}-1$ (Case (2)), we have
an expression of the form:
\begin{eqnarray}
F_{\bm}(t)&=&\sum_{\ell_1,...,\ell_{2k}\in \Z_+} A_{m_1,...,m_k}(\ell_1,...,\ell_{2k})\nonumber \\
&&\qquad \times \big((J_{k+1}-\epsilon_k)\hy_{2k}\big)^{\ell_{2k} }
(ty_{2k-1})^{\ell_{2k-1}} \cdots (ty_1)^{\ell_1}\\
\epsilon_k&=&\delta_{m_{k-1},m_k+1}
\end{eqnarray}
We now have three possibilities for the Motzkin path:

Case (a): $m_{k+1}=m_k$.

Case (b): $m_{k+1}=m_k-1$.

Case (c): $m_{k+j}=m_k+j$ for $j=0,1,2,...,\ell-1$ and $m_\ell\leq m_{\ell-1}$.

Writing $J_{k+1}=(1-\hy_{2k+1}-J_{k+2}\hy_{2k+2})^{-1}$, we note that we always have
$\hy_{2k} \hy_{2k+1}=p\hy_{2k+1}\hy_{2k}$, $\hy_{2k} \hy_{2k+2}=p\hy_{2k+2}\hy_{2k}$, while
$\hy_{2k}$ commutes with $J_{k+2}$, as the latter only depends on $y_i$ with $i\geq 2k+3$.
This allows to write:
\begin{eqnarray}
{\rm Case (1)}\!\!:&\!\!\!\! (J_{k+1}\hy_{2k})^{m}&\!\!\!=\phi_m(\hy_{2k+1}+J_{k+2}\hy_{2k+2}) (t y_{2k})^m \nonumber \\
{\rm Case (2)}\!\!:&\!\!\!\!((J_{k+1}-1)\hy_{2k})^{m}&\!\!\!=p^{m(m-1)\over 2} \phi_m(\hy_{2k+1}+J_{k+2}\hy_{2k+2}) 
(\hy_{2k+1}+J_{k+2}\hy_{2k+2})^m (\hy_{2k})^m\nonumber \\
&  &\!\!\!= \phi_m(\hy_{2k+1}+J_{k+2}\hy_{2k+2}) 
(\hy_{2k+1}+J_{k+2}\hy_{2k+2})^m (ty_{2k+1})^{-m} (ty_{2k})^m\label{cases}
\end{eqnarray}
where we have used  $\hy_{2k}=t y_{2k}$ in the case (1), and 
$\hy_{2k}=y_{2k+1}^{-1}y_{2k}$ in the case (2), hence
$(\hy_{2k})^m =p^{-{m(m-1)\over 2}}(ty_{2k+1})^{-m}(ty_{2k})^m$.

We first treat the Cases (a) and (b).
In both cases, we have $y_{2k+1}\hy_{2k+2}=p\, \hy_{2k+2}y_{2k+1}$, while $y_{2k+1}$ 
commutes with $J_{k+2}$. This suggests to rewrite
$$ \hy_{2k+1}+J_{k+2}\hy_{2k+2} =ty_{2k+1}+(J_{k+2}-\epsilon_{k+1})\hy_{2k+2}, 
\qquad \epsilon_{k+1}=\delta_{m_{k+1},m_k-1}$$
in which the two summands $p$-commute.
In the Cases (1a) and (1b) we get:
$$\phi_m(ty_{2k+1}+(J_{k+2}-\epsilon_k)\hy_{2k+2})=
\sum_{j,n\geq 0} \left[\begin{matrix}m-1+j+n\\ 
m-1,j,n\end{matrix}\right]_p ((J_{k+2}-\epsilon_k)\hy_{2k+2})^n(ty_{2k+1})^j
$$
while in the Cases (2a) and (2b):
\begin{eqnarray*}
&&\!\!\!\!\!\!\!\!\!\!\!\!\!\!\!\!\!\!\!\!\!\!\!\!\!\!\!\!\!\!\!\!\!\!\!\!\!\!\!\!\!\!
\phi_m(ty_{2k+1}+(J_{k+2}-\epsilon_k)\hy_{2k+2}) (ty_{2k+1}+(J_{k+2}-\epsilon_k)\hy_{2k+2})^m\\
&&=\sum_{j\geq 0} \left[\begin{matrix}m-1+j\\ 
j\end{matrix}\right]_p (ty_{2k+1}+(J_{k+2}-\epsilon_k)\hy_{2k+2})^{m+j}\\
&&= \sum_{j,n\geq 0\atop j+n\geq m}
\left[\begin{matrix}j+n-1\\ 
m-1\end{matrix}\right]_p
\left[\begin{matrix}j+n\\ 
j\end{matrix}\right]_p
((J_{k+2}-\epsilon_k)\hy_{2k+2})^n(ty_{2k+1})^j
\end{eqnarray*}

We now deal with the Case (c). By Theorem \ref{mixSJ}, the strictly ascending segment of length $\ell$
starting at $m_k$ corresponds to the mixed Stieltjes-Jacobi expression:
\begin{eqnarray*}
J_{k+1}&=&(1-t y_{2k+1}- S_{2k+2}ty_{2k+2})^{-1}\\
S_i&=& (1-S_{i+1}ty_{i+1})^{-1} \qquad \qquad (i\in [2k+2,2k+2\ell-2])\\
S_{2k+2\ell-1}&=& (1-(J_{k+\ell+1}-\epsilon_{k+\ell})\hy_{2k+2\ell})^{-1}
\end{eqnarray*}
Eq.\eqref{cases} may be rephrased in the present case by substituting
$\hy_{2k+1}+J_{k+2}\hy_{2k}$ with $t y_{2k+1}+S_{2k+2}ty_{2k+2}$.
As before the two cases (1) and (2) lead respectively to:
\begin{eqnarray*}
&&\phi_m(ty_{2k+1}+(S_{2k+2})ty_{2k+2})=
\sum_{j,n\geq 0} \left[\begin{matrix}m-1+j+n\\ 
m-1,j,n\end{matrix}\right]_p (S_{2k+2}ty_{2k+2})^n(ty_{2k+1})^j\\
&&
\phi_m(ty_{2k+1}+(S_{2k+2})ty_{2k+2})(ty_{2k+1}+(S_{2k+2})ty_{2k+2})^m\\
&&\qquad \qquad \qquad=
\sum_{j,n\geq 0\atop j+n\geq m}
\left[\begin{matrix}j+n-1\\ 
m-1\end{matrix}\right]_p
\left[\begin{matrix}j+n\\ 
j\end{matrix}\right]_p
(S_{2k+2}ty_{2k+2})^n(ty_{2k+1})^j
\end{eqnarray*}

Due to the commutation relations between the $y$'s involved, we may apply directly the
recursion in the proof of
Theorem \ref{monethm}, to get:
\begin{eqnarray*}
(S_{2k+2}ty_{2k+2})^n&=&\sum_{\ell_{2k+3},\ell_{2k+4},...,\ell_{2k+2\ell}\geq 0\atop \ell_{2k+2}=n} \prod_{i=2k+2}^{2k+2\ell-1} 
\left[\begin{matrix}\ell_{i}+\ell_{i+1}-1\\ 
\ell_{i+1}\end{matrix}\right]_p \\
&&\times\ ((J_{k+\ell+1}-\epsilon_{k+\ell})\hy_{2k+2\ell})^{\ell_{2k+2\ell}}
(ty_{2k+2\ell-1})^{\ell_{2k+2\ell-1}} \cdots (ty_{2k+3})^{\ell_{2k+3}}(ty_{2k+2})^{n}
\end{eqnarray*}

Assembling all the cases, we find the following recursion relation for $A$ in the cases ($i$a) and ($i$b),
with $i=1,2$:
\begin{eqnarray}
&&A_{m_1,...,m_{k+1}}(\ell_1,...,\ell_{2k+2})=
A_{m_1,...,m_{k}}(\ell_1,...,\ell_{2k}) U_i(\ell_{2k},\ell_{2k+1},\ell_{2k+2}) \nonumber \\
&&\qquad \left\{\begin{matrix}
U_1= \left[ \begin{matrix} \ell_{2k}+\ell_{2k+1}+\ell_{2k+2}-1\\\ell_{2k}-1, \ell_{2k+1},\ell_{2k+2}\end{matrix}\right]_p\\
U_2=  \left[ \begin{matrix} \ell_{2k}+\ell_{2k+1}+\ell_{2k+2}-1\\\ell_{2k}-1\end{matrix}\right]_p
   \left[ \begin{matrix}\ell_{2k}+ \ell_{2k+1}+\ell_{2k+2}\\\ell_{2k+2}\end{matrix}\right]_p
\end{matrix} \right.
\end{eqnarray}
and in the cases ($i$c), $i=1,2$:
\begin{eqnarray}
&&\quad A_{m_1,...,m_{k+\ell-1}}(\ell_1,...,\ell_{2k+2\ell})\\
&&\qquad\qquad\qquad=
 A_{m_1,...,m_{k}}(\ell_1,...,\ell_{2k}) 
U_i(\ell_{2k},\ell_{2k+1},\ell_{2k+2}) V(\ell_{2k+2},...,\ell_{2k+2\ell})  \nonumber \\
&&
\quad V(\ell_{2k+2},...,\ell_{2k+2\ell})= \prod_{i=2k+2}^{2k+2\ell-1} 
\left[\begin{matrix}\ell_{i}+\ell_{i+1}-1\\ 
\ell_{i+1}\end{matrix}\right]_p
\end{eqnarray}

This determines $A$ entirely, with the initial conditions:
\begin{eqnarray*}
{\rm Case \,(a)\, or\, (b)}:& A_{m_1}(\ell_1,\ell_2)=& \left[\begin{matrix}\ell_{1}+\ell_{2}-1\\ 
\ell_{2}\end{matrix}\right]_p
\\
{\rm Case (c)}:& A_{m_1,m_2,...,m_{\ell}}(t)=& \prod_{i=1}^{\ell-1}\left[\begin{matrix}\ell_{i}+\ell_{i+1}-1\\ 
\ell_{i+1}\end{matrix}\right]_p
\end{eqnarray*}
and the Theorem follows,
with $A_\bm(\ell_1,...,\ell_{2r+1})=A_{m_1,...,m_r,m_r}(\ell_1,...,\ell_{2r+1},0)$, as the last step
has $y_{2r+2}=0$.
\end{proof}

Extracting the coefficient of $t^n$ in the series of Theorem \ref{toe}, and using the expressions for
the weights given by Theorem \ref{ysol}, we immediately deduce:
\begin{cor}
The solution $R_{1,n}$ of the quantum $A_r$ $Q$-system is expressed as a Laurent polynomial of
any admissible initial data, with coefficients in $\Z_+[q,q^{-1}]$.
\end{cor}

%\medskip
%\noindent{\bf Acknoledgments.} We would like to thank R. Kedem for numerous discussions.
%We also thank S. Fomin for hospitality at the University of Michigan
%and many discussions while this work was completed.
%We received partial support from the ANR Grant GranMa, the
%ENIGMA research training network MRTN-CT-2004-5652,
%and the ESF program MISGAM. 

%\begin{example}
%\end{example}

\end{document}